\documentclass[12pt, reco]{article}
\usepackage{amssymb,amscd,amsmath,amsthm}
\usepackage{graphicx}
\def\a{\alpha}
\def\G{\Gamma}
\newtheorem{theorem}{Theorem}[section]
\newtheorem{proposition}[theorem]{Proposition}
\newtheorem{corollary}[theorem]{Corollary}
\newtheorem{lemma}[theorem]{Lemma}
\newtheorem{definition}[theorem]{Definition}

\newtheorem{remark}[theorem]{Remark}

\def\cb{{\mathcal B}}
\def\�{{\mathcal C}}
\def\Ht{\widehat{H}}

\def\ce{{\mathcal E}}

\def\Sl{{\mathcal S}}

\def\bn{{\mathbb N}}

\def\br{{\mathbb R}}

\def\ot{\otimes}
\def\Ot{\bigotimes}
\def\Og{\Omega}

\def\a{\alpha}

\def\ffi{\varphi}

\def\Lb{\Lambda}
\def\G{\Gamma}
\def\og{\omega}
\def\El{\mathcal E}
\def\Al{\mathcal{A}}
\def\Hl{\mathcal{H}}

\def\L{\Lambda}

\def\Rl{\mathcal{R}}
\def\ffi{\varphi}
\def\Bl{\mathcal{B}}

\def\Tr{\mathrm{Tr}}
\def\<{\langle}
\def\>{\rangle}
\def\1{\mathbf{1}}

\def\cal{\mathcal}

\def\no{\nonumber\\}

\def\s{\sigma}

\def\id{{\bf 1}\!\!{\rm I}}
\begin{document}
 \begin{center}
{\Large {\bf Quantum Markov States on Cayley trees}}\\[1cm]
\end{center}

\begin{center}
{\large {\sc Farrukh Mukhamedov}}\\[2mm]
\textit{ Department of Mathematical Sciences,\\
College of Science, United Arab Emirates University,\\
P.O. Box, 15551, Al Ain, Abu Dhabi, UAE}\\
E-mail: {\tt far75m@yandex.ru, \ farrukh.m@uaeu.ac.ae}
\end{center}

\begin{center}
{\sc Abdessatar Souissi}\\
\textit{$^1$
College of Business Administration,\\
Qassim university, Buraydah, Saudi Arabia}\\
\textit{$^2$
 Preparatory Institute for Scientific and Technical Studies\\
Carthage University, Carthage,   Tunisia}\\
 
E-mail: {\tt s.abdessatar@hotmail.fr,  \  a.souaissi@qu.edu.sa}\\
\end{center}

\begin{abstract}
It is known that any locally faithful quantum Markov state (QMS) on one dimensional setting can be
considered as a Gibbs state associated
with Hamiltonian with commuting nearest-neighbor interactions. In our previous results, we have
investigated quantum
Markov states (QMS) associated with Ising type models with competing interactions, which are
expected to be QMS, but up to now,
there is no any characterization of QMS over trees. We notice that these QMS do not have
one-dimensional analogues, hence results
of related to one dimensional QMS are not applicable. Therefore, the main aim of the present paper is
to describe of QMS over Cayley trees.
Namely, we prove that any QMS (associated with localized conditional expectations) can be realized
as integral of product states w.t.r. a Gibbs measure.
Moreover, it is established that any locally faithful QMS associated with localized conditional
expectations can be considered as a Gibbs state
  corresponding to Hamiltonians (on the Cayley tree) with commuting competing interactions.

\end{abstract}

\section{Introduction }\label{intr}

It is known that \cite{BR},  in quantum statistical mechanics, concrete
systems are identified with states on corresponding algebras. In
many cases, the algebra can be chosen to be a quasi--local algebra
of observables. The states on these algebras satisfying
Kubo--Martin--Schwinger  (KMS) boundary condition, as is known, describe
equilibrium states of the quantum system under consideration. On the
other hand, for classical systems with the finite radius of
interaction, limiting Gibbs measures are known to be Markov random
fields, see e.g. \cite{D, Geor, [Pr]}. In connection with this, there
is a problem of constructing analogues of non commutative Markov
chains, which arises from quantum statistical mechanics and quantum
field theory in a natural way \cite{DM}. This problem was firstly explored in
\cite{[Ac74f]} by introducing non commutative Markov chains on the
algebra of quasi--local observables. The reader is
referred to \cite{AF03}--\cite{AcLi},\cite{fannes,G,ILW,Kum} and the
references cited therein, for recent development of the theory of
quantum stochastic processes and their applications.

The investigation of a particular class of quantum Markov chains,
called \textit{quantum Markov states (QMS)}, was pursued in \cite{AF03, AcLi,AW}, where
connections with properties of the modular operator of the states
under consideration were established \cite{[AcFr80],GZ}. This provides natural
applications to temperature states arising from suitable quantum
spin models, that is natural connections with the KMS boundary
condition.\footnote{Most of the states arising from Markov processes
considered in \cite{aklt,fannes,fannes2} describe ground states (i.e. states at zero
temperature) of suitable models of quantum spin chains.}

In \cite{AF03}, the most general one dimensional quantum Markov
states have been considered. Among the other results concerning the
structure of such states, a connection with classes of local
Hamiltonians satisfying certain commutation relations and quantum
Markov states has been obtained. The situation arising from quantum
Markov states on the chain, describes one dimensional models of
statistical mechanics with mutually commuting nearest neighbor
interactions. Namely, one dimensional quantum Markov states are very
near to be (diagonal liftings of) ``Ising type'' models, apart from
noncommuting boundary terms.

One of the basic open problems in quantum probability is the
construction of a theory of quantum Markov fields, that are quantum
processes with multi-dimensional index set \cite{[AcFi01b]}. This program concerns
the generalization of the theory of Markov fields (see
\cite{D},\cite{Geor})) to a non-commutative setting, naturally
arising in quantum statistical mechanics and quantum field theory. 

First attempts to construct quantum analogues of
classical Markov fields have been done in
\cite{[AcFi01a]}-\cite{[AcFi01b]},\cite{AcLi,[Liebs99]}. In
these papers the notion of {\it quantum Markov state}, introduced in
\cite{[AcFr80]}, extended to fields as a sub-class of the quantum
Markov chains. In \cite{[AcFiMu07]} a more general definition of
quantum Markov states and chains, including all the presently known
examples, have been extended. In the mentioned papers
quantum Markov fields were considered over multidimensional integer
lattice $\mathbb{Z}$ which, due to the existence of loops, did not allow to construct explicit
examples of such kind of fields. It is known \cite{Bax,[Sp75]} that explicit Gibbs measures can be
obtained on regular trees, therefore, in \cite{AOM,MBS162}, quantum Markov fields (or quantum Markov chains
(QMC)) has been constructed over such trees. Moreover, certain concrete examples were provided. This
direction opened a new direction in the study and construction of QMC via investigation of lattice
models on trees \cite{AMSa0}-\cite{AMSa3},\cite{MBS161}. Mostly, the existing works based on certain models over
the Cayley trees (or Bathe lattices) \cite{Ost}. In fact, even if several definitions of
quantum Markov fields on trees (and more generally on graphs) have been proposed, a really
satisfactory, general theory is still missing and physically interesting examples of such fields
in dimension $d\geq 2$ are very few.

On the other hand, taking into account results of \cite{AF03} any QMS (one dimensional setting) can
be considered as Gibbs state associated with Hamiltonian with commuting nearest-neighbor
interactions. The models considered in \cite{AMSa3, MBS161,MBS162,MR1} satisfy this type of
condition, and hence, roughly speaking, QMC considered there, are expected to be QMS, but up to now,
there is no any characterization of QMS over trees. On the other hand, those QMC do not have
one-dimensional analogues, hence results of \cite{AF03} are not applicable. Therefore, main aim of
the present paper, is to fill this gap, i.e. we are going to describe of QMS over Cayley trees. This
will allow us, in our further investigations, to distinguish which QMC may satisfy KMS conditions
(see \cite{AcLi,GZ}).

We emphasize that any notion of Markovianity strongly depends on an
underlying notion of localization and it is known that, both in the classical
and the quantum case, if the localization is sufficiently rough, then any state
can be considered as Markov chain. Therefore, if one considers the 
the localization given by the levels of the tree, a Markov
field simply becomes a non-homogeneous Markov chain and in this case the
structure of the subclass of Markov states is known \cite{AF03} .

In this paper, we investigate Markov property not only w.r.t. levels of the Cayley tree but also w.r.t. its finer localization structure of the considered tree through 
considering suitable  quasi-conditional expectation called {\it localized},
which keeps into account this finer localization and to prove the structure
theorem corresponding to this localization. An interesting consequence of this structure theorem is that the notion of competing interactions, previously introduced by hands \cite{MBS161,MBS162}, 
now emerges as a consequence of the intrinsic denition combined with the structure theorem.
Therefore, the present paper's main result differs  from the non-homogeneous one dimensional case studied in  \cite{AF03}.  

Notice in our previous work \cite{AMS} we investigated Markov property on the finer structure  of general graphs for backward quantum Markov fields through a specific tessellation on the set of vertices.

Let us outline the organization of the paper. After preliminary
information (see Section 2), in Section 3 we recall definition of quantum Markov chains and states
on Cayley trees. In Section 4, localized conditional expectations (connected to the tree) are
considered and described. In section 5, a Gibbs measure (see \cite{Geor,Roz} for Gibbs states on
Cayley trees) is constructed by means of QMS associated with localized conditional expectations. In
section 6, using the results of sections 4 and 5, we prove that any QMS (associated with localized
conditional expectations) can be realized as integral of product states w.t.r. the Gibbs measure. In
section 7, we prove a reconstruction result. Finally, in section 8, we will establish that any
locally faithful QMS associated with localized conditional expectations can be considered as Gibbs
state corresponding to Hamiltonians (on the Cayley tree) with commuting competing interactions which
implies that all QMC considered in \cite{MBS161,MBS162} are indeed QMS.

\section{Cayley tree}

Let $\Gamma^k_+ = (L,E)$ be a semi-infinite Cayley tree of order
$k\geq 1$ with the root $x^0$ (i.e. each vertex of $\Gamma^k_+$ has
exactly $k+1$ edges, except for the root $x^0$, which has $k$
edges). Here $L$ is the set of vertices and $E$ is the set of edges.
The vertices $x$ and $y$ are called {\it nearest neighbors} and they
are denoted by $l=<x,y>$ if there exists an edge connecting them. A
collection of the pairs $<x,x_1>,\dots,<x_{d-1},y>$ is called a {\it
path} from the point $x$ to the point $y$. The distance $d(x,y),
x,y\in V$, on the Cayley tree, is the length of the shortest path
from $x$ to $y$.

Let us set
\[
W_n = \{ x\in L \, : \, d(x,x_0) = n\} , \qquad \Lambda_n =
\bigcup_{k=0}^n W_k, \qquad  \L_{[n,m]}=\bigcup_{k=n}^mW_k, \ (n<m)
\]
\[
E_n = \big\{ <x,y> \in E \, : \, x,y \in \Lambda_n\big\}, \qquad
\Lambda_n^c = \bigcup_{k=n}^\infty W_k
\]

Recall a coordinate structure in $\G^k_+$:  every vertex $x$ (except
for $x^0$) of $\G^k_+$ has coordinates $(i_1,\dots,i_n)$, here
$i_m\in\{1,\dots,k\}$, $1\leq m\leq n$ and for the vertex $x^0$ we
put $(0)$.  Namely, the symbol $(0)$ constitutes level 0, and the
sites $(i_1,\dots,i_n)$ form level $n$ (i.e. $d(x^0,x)=n$) of the
lattice.


For $x\in \G^k_+$, $x=(i_1,\dots,i_n)$ denote
$$ S(x)=\{(x,i):\ 1\leq
i\leq k\}.
$$
Here $(x,i)$ means that $(i_1,\dots,i_n,i)$. This set is called a
set of {\it direct successors} of $x$.

Two vertices $x,y\in V$ is called {\it one level
next-nearest-neighbor  vertices} if there is a vertex $z\in V$ such
that  $x,y\in S(z)$, and they are denoted by $>x,y<$. In this case
the vertices $x,z,y$ is called {\it ternary} and denoted by
$<x,z,y>$.

Let us rewrite the elements of $W_n$ in the following lexicographic order (w.r.t. the
coordinate system)
\begin{eqnarray*}
\overrightarrow{W_n}:=\left(x^{(1)}_{W_n},x^{(2)}_{W_n},\cdots,x^{(|W_n|)}_{W_n}\right).
\end{eqnarray*}
Note that $|W_n|=k^n$. In this lexicographic order, the vertices
$x^{(1)}_{W_n},x^{(2)}_{W_n},\cdots,x^{(|W_n|)}_{W_n}$ of $W_n$ can
be represented in terms of the coordinate system as follows
\begin{eqnarray}\label{xw}
&&x^{(1)}_{W_n}=(1,1,\cdots,1,1), \quad x^{(2)}_{W_n}=(1,1,\cdots,1,2), \ \ \cdots \quad
x^{(k)}_{W_n}=(1,1,\cdots,1,k,),\\
&&x^{(k+1)}_{W_n}=(1,1,\cdots,2,1), \quad
x^{(2)}_{W_n}=(1,1,\cdots,2,2), \ \ \cdots \quad
x^{(2k)}_{W_n}=(1,1,\cdots,2,k),\nonumber
\end{eqnarray}
\[\vdots\]
\begin{eqnarray*}
&&x^{(|W_n|-k+1)}_{W_n}=(k,k,,\cdots,k,1), \
x^{(|W_n|-k+2)}_{W_n}=(k,k,\cdots,k,2),\ \ \cdots
x^{|W_n|}_{W_n}=(k,k,\cdots,k,k).
\end{eqnarray*}

Analogously, for a given vertex $x,$ we shall use the following
notation for the set of direct successors of $x$:
\begin{eqnarray*}
\overrightarrow{S(x)}:=\left((x,1),(x,2),\cdots (x,k)\right),\quad
\overleftarrow{S(x)}:=\left((x,k),(x,k-1),\cdots (x,1)\right).
\end{eqnarray*}

%
%
%

\section{Quantum Markov chains and states}

 The algebra of observables $\cb_x$ for any single site
$x\in L$ will be taken as the algebra $M_d$ of the complex $d\times
d$ matrices. The algebra of observables localized in the finite
volume $\L\subset L$ is then given by
$\cb_\L=\bigotimes\limits_{x\in\L}\cb_x$. As usual if
$\L^1\subset\L^2\subset L$, then $\cb_{\L^1}$ is identified as a
subalgebra of $\cb_{\L^2}$ by tensoring with unit matrices on the
sites $x\in\L^2\setminus\L^1$. Note that, in the sequel, by
$\cb_{\L,+}$ we denote the set of all positive elements of $\cb_\L$
(note that an element is positive if its spectrum is located in
$\br_+$). The full algebra $\cb_L$ of the tree is obtained in the
usual manner by an inductive limit
$$
\cb_L=\overline{\bigcup\limits_{\L_n}\cb_{\L_n}}.
$$

In what follows, by ${\cal S}({\cal B}_\L)$ we will denote the set
of all states defined on the algebra ${\cal B}_\L$.

Consider a triplet ${\cal C} \subset {\cal B} \subset {\cal A}$ of
unital $C^*$-algebras. Recall \cite{ACe} that a {\it
quasi-conditional expectation} with respect to the given triplet is
a completely positive (CP), identity-preserving linear map $\ce
\,:\, {\cal A} \to {\cal B}$ such that $ \ce(ca) = c \ce(a)$, for
all $a\in {\cal A},\, c \in {\cal C}$.

In what follows, by (Umegaki) \textit{conditional expectation}
$E:\mathcal{A}\to \mathcal{B}$ we mean a norm-one projection of the
$C^*$-algebra $\mathcal{A}$ onto a $C^*$-subalgebra (with the same
identity $\id$) $\mathcal{B}$. The map $E$ is automatically a
completely positive, identity-preserving $\mathcal{B}$-module map
\cite{Str}. If $\mathcal{A}$ is a matrix algebra, then the structure
of a conditional expectation is well-known \cite{AcLi}. Let is
recall some facts. Assume that $\mathcal{A}$ is a full matrix
algebra, and consider the (finite) set of $\{P_i\}$ of minimal
central projections of the range $\mathcal{B}$  of $E$, we have
$$
E(x)=\sum_iE(P_ixP_i)P_i.
$$
Then $E$ is uniquely determined by its values on the reduced
algebras
$$
\mathcal{A}_{P_i}:=P_i\mathcal{A}P_i=N_i\otimes\bar N_i,
$$
where $N_i\sim \mathcal{B}_{P_i}:=\mathcal{B}P_i$ and $\bar
N_i:=\mathcal{B}'P_i$ (here the commutant $\mathcal{B}'$ is
considered relative to $\mathcal{A}$). Moreover, there exist states
$\phi_i$ on $\bar N_i$ such that
\begin{equation}\label{Exp1}
E(P_i(a\otimes \bar a)P_i)=\phi_i(\bar a)P_i(a\otimes \id)P_i.
\end{equation}
For the general theory of operator algebras we refer to
\cite{BR,Str}.

\begin{definition}\cite{[AcFiMu07],AOM}\label{QMCdef}
Let $\varphi$ be a state on ${\cal B}_L$. Then $\ffi$ is called a
{\it (backward) quantum Markov chain}, associated to $\{\L_n\}$, if
there exists a quasi-conditional expectation $\ce_{\Lambda_n}$ with
respect to the triple ${\cal B}_{{\Lambda}_{n-1}}\subseteq {\cal
B}_{\Lambda_n}\subseteq{\cal B}_{\Lambda_{n+1}}$ for each
$n\in\bn$and an initial state $\rho_0\in S(B_{\L_0})$ such that
\begin{equation*}
\varphi = \lim_{n\to\infty} \rho_0\circ \ce_{\Lambda_0}\circ
\ce_{\Lambda_{1}} \circ \cdots \circ \ce_{\Lambda_n}
\end{equation*}
in the weak-* topology.
\end{definition}

 \begin{definition}\label{QMS1}\cite{[AcFiMu07]}
A quantum Markov chain $\ffi $ is said to be quantum Markov state
with respect to the sequence $\{\ce_{\Lambda_j}\}$ of
quasi-conditional expectations if one has
$$\ffi_{\lceil \Bl_{\Lambda_j}}\circ \ce_{\Lambda_j}=\ffi_{\lceil \Bl_{\Lambda_{j+1}}}, \ \
j\in\bn.$$
\end{definition}

In what follows, we always assume that the states are {\it locally
faithful} (i.e. states on $\Bl$ with faithful restrictions to local
subalgebras). By the standard way (see \cite{[AcFi01b],[AcFiMu07]})
one can show that the Markov property defined above can be stated by
a sequence of global quasi--conditional expectations, or equally
well by sequences of local or global conditional expectations. By
putting $e_{n}:=\ce_{\Lambda_n}\lceil_{\Bl_{\Lambda_{[n,n+1]}}}$, it
will be enough to consider the ergodic limits
$$
\ce^{(n)}:=\lim_m\frac{1}{m}\sum^{m-1}_{h=0}(e_{n})^{h}\,,
$$
which give rise to a sequence of two--step conditional expectations,
called {\it transition expectations} in the sequel.

For $j>0$, we define the conditional expectation $E_{j}$ from $\Bl_{\Lb_{j+1}}$ into $\Bl_{\Lb_{j}}$
by:
$$E_{j}\left(a_{W_0}\ot\cdots \ot a_{W_j}\ot a_{W_{j+1}}\right)=a_{W_0}\ot\cdots \ot a_{W_{j-1}}\ot
\El^{(j)}\left(a_{W_{j}}\ot a_{W_{j+1}}\right)$$

 One can prove
the following

\begin{proposition}
Let $\varphi$ be a state on the $\Bl$. The following assertions are
equivalent.
\begin{itemize}
\item[(i)] $\varphi$ is a quantum Markov state;
\item[(ii)] the properties listed in Definitions \ref{QMCdef} and \ref{QMS1} are
satisfied if one replaces the quasi--conditional expectations
$\ce_{\Lambda_n}$ with Umegaki conditional expectations $E_{n}$.
\end{itemize}
\end{proposition}

The next result describes the quantum Markov states. Note that if
the tree $\Gamma^k_+$ is one dimensional (i.e: k=1), then the similar result
has been proven in \cite{AcLi,AF03}.

\begin{theorem}\label{caracofmarkovstate}
Let $\ffi\in \Sl(\Bl_L)$. Then $\ffi$ is a quantum Markov state
w.r.t the sequence of transition expectations
$\{\El^{(j)}\}_{j\geq0}$ if and only if
   \begin{equation}\label{Mp1}
\ffi(a)=\ffi\left(\El^{(0)}\left(a_{W_0}\ot\cdots \ot
\El^{(n-1)}\big(a_{W_{n-1}}\ot\El^{(n)}\left(a_{W_n}\ot a_{W_{n+1}}\right)\cdots\right)\right)
   \end{equation}
for every $n\in \bn$, and $a=a_{W_0}\ot \cdots a_{W_n}\ot a_{W_{n+1}}$ any linear generator of
$\Bl_{\Lb_{n+1}}$,
  with $a_{W_j}\in \Bl_{W_j}$ for $j=1\cdots n+1.$
\end{theorem}

\begin{proof}
Suppose that $\ffi$ is a quantum Markov state w.r.t the sequence
$\{\El^{(j)}\}_{j\geq 0}$.  Then for any $a=a_{W_0}\ot a_{W_1}\ot
\cdots a_{W_n}\ot a_{W_{n+1}}\in \Bl_{\Lambda_{n+1}}$, by means of
the Markov property one gets
$$\ffi(a)=\ffi\left(a_{W_0}\cdots \ot
a_{W_{n-1}}\ot\El^{(n)}\left(a_{W_n}\ot a_{W_{n+1}}\right)\right)$$
Then by repeating the  application of the Markov property $n$ more
times we obtain  \eqref{Mp1}.

Conversely, assume that $\ffi$ satisfies the chain of conditions \eqref{Mp1} 
Then for a fixed \begin{equation}\label{AA}
a=a_{W_0}\ot  \cdots
a_{W_n}\ot a_{W_{n+1}}
\end{equation}
from $\Bl_{\Lb_{n+1}}$ by \eqref{Mp1} one finds
$$ \ffi(a)=\ffi\left(\El^{(0)}\left(a_{W_0}\ot\cdots \ot
\El^{(n-1)}\big(a_{W_{n-1}}\ot\El^{(n)}\left(a_{W_n}\ot a_{W_{n+1}}\right)\cdots\right)\right)$$
And again by application of \eqref{Mp1} we get 
$$\ffi\left(\El^{(0)}\left(a_{W_0}\ot\cdots \ot \El^{(n-1)}\big(
a_{W_{n-1}}\ot\El^{(n)}\left(a_{W_n}\ot a_{W_{n+1}}\right)\cdots\right)\right)
   =\ffi\left(a_{W_0}\ot\cdots\ot a_{W_{n-1}}\ot \widehat{a}_{W_n}\right)$$
where $\widehat{a}_{W_n}=\El^{(n)}\left(a_{W_n}\ot
a_{W_{n+1}}\right)$. Then keeping mind the equality
$$a_{W_0}\ot\cdots\ot a_{W_{n-1}}\ot \widehat{a}_{W_n}=E_n(a)$$
we obtain the Markov property for $\ffi$
$$\ffi_{\lceil\Bl_{\Lb_{n+1}}}\left(a\right)=\ffi_{\lceil
\Bl_{\Lb_n}}\left(E_n(a)\right),$$ since the elements of the form
\eqref{AA} generate $\Bl_{\Lb_{n+1}}$.
\end{proof}

\section{Localized conditional expectations}

In this section, we consider and describe localized conditional expectations. 
Namely, let $\El^{(j)}$ be a transition expectation from
$\Bl_{\Lambda_{[j,j+1]}}$ to $\Bl_{W_j}$. The transition expectation
$\El^{(j)}$ is called \textit{localized} if one has
\begin{equation}\label{decoE}
\El^{(j)}=\Ot_{x\in W_j}\El_{x}^{(j)},
\end{equation}
where $ \El_{x}^{(j)} :\Bl_{x}\ot \Bl_{(x,1)}\ot\cdots\ot
\Bl_{(x,k)} \rightarrow  \Bl_x$ is a conditional expectation.

\begin{remark} We notice that if one considers conditional expectations without localization
property, then the results of \cite{AF03} can be applied to the considered QMS and one can get the
disintegration of QMS, which would be not enough for its finer representation. Roughly speaking, in
that case, the Hamiltonians (see Section 8) would be defined on the levels $W_j$ and their structure
would not be described. Therefore, for our need, we have to impose the localization, which would
yield a desired representation of QMS.
\end{remark}

In what follows, we will use techniques of \cite{AcLi,FI} related to conditional expectations. 

For each $j\geq 0$ and $x\in W_j$ by $\mathcal{R}(\El^{(j)}_{x})$ we
denote the range of the transition expectation $\El_x^{(j)}$. Consider its centeral 
$Z^j_x$ (i.e. the center of $\mathcal{R}(\El^{j}_{x})$), together with its spectrum $\Omega^{j}_x$,
and the set of minimal central
projections $\{P^j_{\omega_x}\}_{\omega_x\in \Omega_{x}}$ (which is finite whenever
$\mathcal{R}(\El^{(j)}_{x})$ is finite dimensional). For the tail of simplicity, central projections $P^j_{\omega_x}$ will be denoted  $P_{\omega_x}$ henceforth.
We have $$\sum_{\Omega_x}P_{\omega_x}=1$$
 We define for $j\geq 0$ and $x\in W_j$.
$$R^{j}_{x}:=\bigoplus _{\Omega_x} P_{\omega_x} \Bl_x P_{\omega_x} $$
And  define the C*-subalgebra $\Rl$ of $\Bl_L$ by:
\begin{equation}\label{otRjx}
\Rl:= \Ot_{j\geq 0}\left(\Ot_{x\in W_j} R^{j}_{x}\right)
\end{equation}
Then we obtain, in a canonical way, a conditional
expectation
$$
E:\Bl_L\mapsto \Rl
$$
 defined to be the (infinite)
tensor product of the following conditional expectations:
\begin{equation}\label{E}
a\in \Bl_x\mapsto \sum_{\omega_x\in \Omega_{x}}P_{\omega_x}a P_{\omega_x}
\end{equation}
Note that for the algebra $R^{j}_x$ one has that
$P_{\omega_x}R^{j}_{x} P_{\omega_x}$ is a factor of
$B(H_{\omega_x})$ with $H_{\omega_x}=P_{\omega_x}H.$

Then $P_{\omega_x}R^{j}_{x} P_{\omega_x}=P_{\omega_x}B(H_{\omega_{x,0}})\Ot \id_{H_{\omega_{x,1}}}
P_{\omega_x}$,
 where $H_{\omega_x}=H_{\omega_{x,0}}\Ot H_{\omega_{x,1}}.$
Using this fact and the fact that the family of central projections
$\{P_{j_x}\}$ is orthogonal, we obtain
\begin{eqnarray}
R^{j}_{x}&=& \sum_{\omega_x\in \Omega_x}P_{\omega_x}R^{j}_{x}P_{\omega_x}\no
&=& \bigoplus_{\omega_x\in \Omega_x}P_{\omega_x}B(H_{\omega_{x,0}})\ot \id_{H_{\omega_{x,1}}}
P_{\omega_x}
\end{eqnarray}
and
$${R^{j}_{x}}'=\bigoplus_{\Omega_x}P_{\omega_x}\id_{H_{\omega_{x,0}}}\ot
B(H_{\omega_{x,1}})P_{\omega_x} $$

Then the transition expectation $\El^{(j)}_{x}$ can be written in a
canonical form: for $a\in \Bl_x^j\Ot
\Bl^{j+1}_{\overrightarrow{S(x)}}$ one has
\begin{equation}\label{decoEj}
\El_{x}^{(j)}(a)=\sum_{\Omega_x}P_{\omega_x}\Phi_{\omega_x}(P_{\omega_x} a P_{\omega_x})P_{\omega_x}\end{equation}
where $\Phi_{\omega_x} :B(H_{\omega_{x,0}})\Ot
B(H_{\omega_{x,1}})\Ot
\Bl_{\overrightarrow{S(x)}}^{(j+1)}\longrightarrow
B(H_{\omega_{x,0}})\Ot \id_{\omega_{x,1}}$ is the Umegaki
conditional expectation defined by a state $\phi_{\omega_x}$ on the
algebra $B(H_{\omega_{x,1}})\Ot \Bl_{\overrightarrow{S(x)}}^{(j+1)}$
in the following way:
$$\Phi_{\omega_x}(a_{\omega_{x,0}}\ot a_{\omega_{x,1}}\ot b)=\phi_{\omega_x}( a_{\omega_{x,1}}\ot
b)a_{\omega_{x,0}}\ot \id_{\omega_{x,1}}$$

Denote ${\Omega}^{(j)}:=\prod_{{x\in
\overrightarrow{W_j}}}\Omega_{x}$, and for
$\omega^j=\left(\omega_{x}\right)_{x\in \overrightarrow{W_j}}\in
{\Omega}^{(j)}$, we define a projection $P_{\omega^{(j)}}$ defined
by:
$$P_{\omega^{(j)}}:=\Ot_{x\in \overrightarrow{W_j}}P_{\omega_x}.$$

\begin{theorem}
Let $j\in \mathbf{N}$ and let for $\omega^{(j)}=(\omega_{x})_{x\in \overrightarrow{W_j}}$,
$$H_{\omega^{(j)},0}:=\Ot_{x\in \overrightarrow{W_j}}H_{\omega_{x,0}}, \quad
H_{\omega^{(j)},1}:=\Ot_{x\in
\overrightarrow{W_j}}H_{\omega_{x,1}}.$$ Then
\begin{enumerate}
  \item[(i)]  The range $R^{(j)}$ of the transition expectation $\El^{(j)}$ satisfies
  \begin{equation}
R^{(j)}\cong \bigoplus_{\omega^{(j)}\in
\Omega^{(j)}}\left(P_{\omega^{(j)}}B(H_{\omega^{(j)},0})\right)\Ot\left(\id_{H_{\omega^{(j)},1}}
P_{\omega^{(j)}}\right)
  \end{equation}
  \item[(ii)] Furthermore, for $a\in\Bl_{W_j}\Ot \Bl_{W_{j+1}}$ one has
\begin{equation}
\El^{(j)}(a)=\sum_{\omega^j\in \Omega^{(j)}}P_{\omega^{(j)}}\Phi_{\omega^{(j)}}(P_{\omega^{(j)}} a
P_{\omega^{(j)}})P_{\omega^{(j)}}
\end{equation}
where  $\Phi_{\omega^{(j)}}= \Ot_{x\in W_j}\Phi_{\omega_x}.$
\end{enumerate}
\end{theorem}

\begin{proof} (i) Due to $\El^{(j)}=\Ot_{x\in \overrightarrow{W_j}}\El_{x}^{(j)}$ the rang of
$\El^{(j)}$ is given by:
\begin{eqnarray}
R^{(j)}&=&\Ot_{x\in \overrightarrow{W_j}}R^{(j)}_{x}\no
&=&\Ot_{x\in \overrightarrow{W_j}}\left[\bigoplus_{\omega_x\in
\Omega_{x}}P_{\omega_x}B(H_{\omega_{x,0}})\ot \id_{H_{\omega_{x,1}}} P_{\omega_x}\right]\no
&=&\bigoplus_{\{\omega_{x_{W_j}(k)}\in \Omega_{x_{W_j}(k)};
k=1,\cdots,|W_j|\}}\left(P_{\omega_{x_{W_j}(1)}}B(H_{\omega_{x_{W_j}(1),0}})\ot
\id_{H_{\omega_{{x_{W_j}(1)},1}}} P_{\omega_{x_{W_j}(1)}}\right)\ot\no
&&\cdots \ot \left(P_{\omega_{x_{W_j}(|W_j|)}}B(H_{\omega_{{x_{W_j}(|W_j|)},0}})\ot
\id_{H_{\omega_{{x_{W_j}(|W_j|)},1}}} P_{\omega_{x_{W_j}(|W_j|)}}\right).\no
&\cong& \bigoplus_{\{\omega_{x_{W_j}(k)}\in \Omega_{x_{W_j}(k)};
k=1,\cdots,|W_j|\}}\left(\Ot_{k=1}^{|W_j|}
P_{\omega_{x_{W_j}(k)}}B(H_{\omega_{{x_{W_j}(k)},0}})\right)\ot\left(\Ot_{k=1}^{|W_j|}
\id_{H_{\omega_{{x_{W_j}(k)},1}}} P_{\omega_{x_{W_j}(k)}}\right)\no
&\cong& \bigoplus_{\omega^{(j)}\in
\Omega^{(j)}}\left(P_{\omega^{(j)}}B(H_{\omega^{(j)},0})\right)\Ot\left(\id_{H_{\omega^{(j)},1}}
P_{\omega^{(j)}}\right)
\end{eqnarray}

(ii) Let $a=a_0^{(j)}\ot a_1^{(j)}\ot b_{j+1}\in \Bl_{W_{n}}\Ot \Bl_{W_{n+1}}$ be a localized
element with
$$a_0^{(j)}=\bigotimes_{x\in W_j}a_{{x,0}}\in B(H_{0}^{(j)}),\quad a_1^{(j)}=\bigotimes_{x\in
W_j}a_{{x,1}}\in B(H_{1}^{(j)}),
\quad b_{j+1}=\bigotimes_{y\in W_{j+1}}b_{y}\in \Bl_{W_{j+1}}.$$
Then we have  
\begin{eqnarray*}
\El^{(j)}(a)&=&\Ot_{x\in W_j}\El_{x}^{(j)}\left(a_{{x,0}}\ot a_{{x,1}}\ot
b_{\overrightarrow{S(x)}}\right)\no
&=&\Ot_{x\in W_j}\bigl[\sum_{\omega_x}P_{\omega_x}\Phi_{\omega_x}(P_{\omega_x} a
P_{\omega_x})P_{\omega_x}\bigr]\no
&=& \sum_{\{\omega_x\in \Omega_{(x)}\mid x\in W_j\}}\Ot_{x\in W_j}
P_{\omega_x}\Phi_{\omega_x}(P_{\omega_x} a P_{\omega_x})P_{\omega_x}\no
&=& \sum_{\omega^{(j)}\in \Omega^{(j)}}P_{\omega^{(j)} }\Phi_{\omega^{(j)} }(P_{\omega^{(j)} } a
P_{\omega^{(j)} })P_{\omega^{(j)} }
\end{eqnarray*}
The last one can be extended to all elements of $\Bl_{W_j}\Ot
\Bl_{W_{j+1}}$.
\end{proof}

\begin{remark}
For ${\omega^{(j)} }=(\omega_x)_{x\in W_j}\in \Omega^{(j)}$ a
product state $\phi_{\omega^{(j)} }$ on $B(H_{{\omega^{(j)} },1})\ot
\Bl_{W_{j+1}}$ is defined on the localized elements by:
\begin{equation}
\phi_{\omega^{(j)} }\left(\Ot_{x\in W_j}(a_{\omega_x,1}\ot
a_{\overrightarrow{S(x)}})\right)=\prod_{x\in W_j}\phi_{\omega_x}(a_{\omega_x,1}\ot
a_{\overrightarrow{S(x)}})
\end{equation}
for every $a_{\omega_x,1}\in B(H_{\omega_x,1})$ and every $a_{\overrightarrow{S(x)}}\in
\Bl_{\overrightarrow{S(x)}}$.
\end{remark}

\section{A classical Gibbs measure}

In this section, we considering a quantum Markov state $\ffi$ on the quasi-local algebra $ \Bl_L$
w.r.t. the
sequence of localized conditional expectations $\{\ce^{(j)}\}$.
Assume, as before, that for every $x\in\G^k(x_0),  \Bl_x \cong \Al$ with $\Al$ is
a fixed finite dimensional C*-algebra. Then the range
$\mathcal{R}(\El^{(j)}_{x})$ of $\El_x^{(j)}$ is finite dimensional.
 The center $Z^{(j)}_{x}$ is finite dimensional and the index set
$\mathcal{I}_x$ is in one to one bijection with the spectrum
$\Omega^{(j)}_{x}:=spec(\mathcal{Z}_{j}^{x})$, then we denote the
set of minimal central projections by $\{P_{\omega_x}\mid
\omega_x\in \Omega_x\}$ and the states $\{\phi_{\omega_x}\mid j_x\in
\Omega_x\}$ by $\{\phi_{\omega_x}\mid \omega_x\in\Omega_x\}$ (see Section 4).

\begin{remark}
\begin{enumerate}
\item Let $\Lb\subset_{fin} L$, if $\psi_x$ is a given state on $\Bl_x$ for all $x\in \Lb$, we
denote the product state $\psi_{\Lb}$ of $\{\psi_i\mid i\in\Lb\}$ defined on $\Bl_\Lb$ by
$\psi_{\Lb}=\otimes_{x\in\Lb}\psi_i.$
\item Let $x\in L$, we denote simply $\psi_{\overrightarrow{S(x)}}:= \ot_{y\in
\overrightarrow{S(x})}\psi_y$ the product state of $\{\psi_y\in \Sl(\Bl_y)\mid y\in
\overrightarrow{S(x)}\}$ and $P_{\psi}=\ot_{y\in \overrightarrow{S(x)}}P_{\psi_y}$ its support
projection for $P_{\psi_y}$ the support projection of $\psi_y$.
\end{enumerate}
\end{remark}
\begin{lemma}
Let $\mathcal{Z}^j$ be the center of $\mathcal{R}(\El^{(j)})$ and $\Omega_j:=spec(\mathcal{Z}^j)$.
Then
\begin{enumerate}
\item $\mathcal{Z}^j=\Ot_{x\in W_j}\mathcal{Z}^j_x.$
\item $\Omega^{(j)}=\prod_{x\in W_j}\Omega^{(j)}_x \quad \left(i.e: \forall \og_j\in \Og_j, \quad
\exists \{\og_x\}_{x\in \overrightarrow{W_j}}) \mid \og^{(j)}=\ot_{x\in \overrightarrow{W_j}}\og_x
\right)$
\end{enumerate}
\end{lemma}

%

For the sake of simplicity, let us denote the spectrum of
$Z:=Z(\Al)$ as follow $spec(Z):=\{\og_1, \og_2,\cdots \og_q\}$ and
the set of associate minimal central projection by
$\{P_{1},P_{2},\cdots,P_{q}\}$.  Without loss of generality we may
suppose that:
$$\forall x\in L, \forall \og_x\in \Og_x, \quad P_{\og_x}=i_{\Lambda_{x,x}}(P_{\og_r}) \quad
\textrm{for some}\quad r\in\{1,2,\cdots,q\} $$
where $i_{\Lambda_{x,x}}$ the canonical embedding of $\Al$ at the $x^{th}$ position on $\Bl_L$.

Let $\Phi:=\{1,2,\cdots,q\}$ we can also take $\Og=\Phi^{L}$
equipped with the $\sigma$-algebra $\mathcal{F}$ generated by its
cylinder subsets. For $A\subset L$ a spin configuration $\sigma_A$
on $A$ is defined as a function $x\in A\mapsto \sigma_A(x)\in \Phi$.
The set of all configurations coincides with $\Og_A=\Phi^{A}$.


\begin{remark}
\begin{enumerate}
\item For $A,B\subset L$, with $A\cap B =\varnothing$, and $\s\in \Og_A, \s'\in \Og_B$, we denote by
$\s\vee\s'$ the configuration defined on $A\cup B$ such that $\s\vee\s'_{\lceil A}=\s$ and
$\s\vee\s'_{\lceil B}=\s'$.
\item Let $j\geq 0$ the set $\Omega^{(j)}$ is equal $\Og_{W_j}$ through the following
identification:
    \begin{eqnarray*}
          \s_{\omega^{(j)}}: W_j &\rightarrow& \Phi \no
             x &\mapsto& \omega_x
       \end{eqnarray*}
If $\s=\s_{\omega^{(j)}}$ we denote the projection $P_{\omega^{(j)}}$ by $P_\s$ and the state
$\phi_{\omega^{(j)}}$ by $\phi_\s$, and for $i\in\{0,1\}$
the Hilbert spaces $H_{\omega_x,i}$ and $H_{{\omega^{(j)}},i}$ will be denoted respectively by
$H_{\s(x),i}$ and $H_{\overrightarrow{S(x)},i}$
\end{enumerate}
\end{remark}
 Now we define by induction the sequence $\{\mu_{\Lb_n}\}_{n\geq 0}$
as follows:
\begin{eqnarray}\label{mu}
\left\{
\begin{array}{ll}
\mu_{\Lb_0}(\s_0)=\ffi\left(i_{\Lb_0}(P_{\og_{\s_0(x_0)}})\right), \quad \forall \s_0\in
\Og_{\Lb_0},\no\\[2mm]
\mu_{\Lb_n}(\s)=\mu_{\Lb_0}(\s_{\lceil
\Lb_0})\prod\limits_{k=1}^{n-1}\prod\limits_{x\in
W_k}\pi_{\og_{\s(x)},\og_{\s(\overrightarrow{S(x)})}}, \quad \forall
\s\in\Og_{\Lb_n}, n\geq1,
\end{array}
\right.
\end{eqnarray}
where $\pi_{\omega_x,\omega_{\overrightarrow{S(x)}}}=
\phi_{\omega_x}\left(\id_{H_{\omega_x,1}} \ot
P_{\omega_{\overrightarrow{S(x)}}}\right)$

\begin{definition}
A sequence $\{P_{\Lb_n}\}_{n\geq 0}$ of probability measures on $\Og$ is said to be compatible if
for all $n\geq 0$ and $\s_n\in \Og_{\Lb_n}$
\begin{equation}
\sum_{\s'\in \Og_{W_{n+1}}}P_{\Lb_{n+1}}(\s_n\vee
\s')=P_{\Lb_n}(\s_n)
\end{equation}
\end{definition}
\begin{proposition}
The sequence $\{\mu_{\Lb_n}\}_{n\geq0}$ is compatible.
\end{proposition}
\begin{proof}
Let $n\geq0$ and $\s\in \Og_{\Lb_n}$. For all $\s'\in \Og_{W_{n+1}}$
one has:
\begin{eqnarray*}
\mu_{\Lb_{n+1}}(\s\vee \s')&=&\mu_{\Lb_0}(\s\vee\s'_{\lceil
\Lb_0})\prod_{k=1}^{n}\prod_{x\in
W_k}\pi_{\og_{\s\vee\s'(x)},\og_{\s\vee\s'(\overrightarrow{S(x)})}}\no
&=&\left(\mu_{\Lb_0}(\s_{\lceil \Lb_0})\prod_{k=1}^{n-1}\prod_{x\in
W_k}\pi_{\og_{\s(x)},\og_{\s(\overrightarrow{S(x)})}}\right)\prod_{x\in
W_n}\pi_{\og_{\s(x)},\og_{\s'(\overrightarrow{S(x)})}}\no &=&
\mu_{\Lb_n}(\s)\prod_{x\in
W_n}\pi_{\og_{\s(x)},\og_{\s'(\overrightarrow{S(x)})}}\no
&=&\mu_{\Lb_n}(\s)\phi_{\s_{\lceil W_n}}(\id_{H_{\s_{\lceil
W_n},1}}\ot P_{\s'})\no
\end{eqnarray*}
According to $\sum_{\s'\in W_{n+1}}P_{\s'}=1$ with
$P_{\s'}=\ot_{x\in  W_{n+1}}P_{\s'(x)}$ we find
\begin{eqnarray*}
\sum_{\s'\in \Og_{W_{n+1}}}\mu_{\Lb_{n+1}}(\s\vee
\s')&=&\mu_{\Lb_n}(\s)\sum_{\s'\in \Og_{W_{n+1}}}\phi_{\s_{\lceil
W_n}}(\id_{H_{\s_{\lceil W_n},1}}\ot P_\s')\no &=&\mu_{\Lb_n}(\s)
.\phi_{\s_{\lceil W_n}}(\id_{H_{\s_{\lceil W_n},1}}\ot(\sum_{\s'\in
W_{n+1}} P_{\s'}))\no &=&\mu_{\Lb_n}(\s).\phi_{\s_{\lceil
W_j}}(\id_{H_{\s\lceil_{W_j},1}}\ot\id_{W_{n+1}})\no &=&
\mu_{\Lb_n}(\s).
\end{eqnarray*}
This completes the proof.
\end{proof}

Due to the previous proposition and to the fact that
${\mu_{\Lb_{n+1}}}_{\lceil \Og_{\Lb_n}}=\mu_{\Lb_n}$, the Kolmogorov
consistency theorem ensures the existence of  a probability measure
$\mu$ on $\Og$ such that $\mu_{\lceil\Og_{\Lb_n}}=\mu_{\Lb_n}$.

Let us consider the Hamiltonian on $\Og_{\Lb_n}$ defined by:
\begin{equation}
\Hl_{\Lb_n}(\s):= \sum_{k=1}^{n-1}\sum_{x\in
W_k}\rho\left(\s(x),\s(\overrightarrow{S(x)})\right)+\rho_0(\s_{\lceil \Lb_0})
\end{equation}
where $\rho_0(\s_0)=\ln (\mu(\s_0))$ for all $\s_0\in \Og_{\Lb_0}$ and
$\rho\left(\s(x),\s(\overrightarrow{S(x)})\right)=\ln\left(\pi_{\og_{\s(x)},\og_{\overrightarrow{S(x)}}}\right)$
for $x\in \G^{k}(x_0).$\\

Then the measure $\mu$ can be viewed as a Gibbs measure for the
Hamiltonian in the following sense:
\begin{equation}
\mu_{\Lb_n}(\s)=\frac{\textrm{e}^{\Hl_{\Lb_n}(\s)}}{Z_n}, \quad
\forall\s\in \Og_{\Lb_n}, \forall n\geq0,
\end{equation}
where $Z_n=\sum_{\s'\in \Og_{\Lb_n}}\textrm{e}^{\Hl_{\Lb_n}(\s')}$
\begin{remark}
\begin{enumerate}
\item The Hamiltonian $\Hl$ is well defined because the state $\ffi$ is locally faithful.
\item In our case $Z_n=\mu_{\Lb_n}(\Og_{\Lb_n})=1.$
\end{enumerate}
\end{remark}

\section{Disintegration of the state $\ffi$ }

We will use the same notations as the previous section. Let $\s\in
\Og$, for $x\in L$ we denote $\Bl_{\s(x)}:=P_{\s(x)}\Bl_xP_{\s(x)}$.
One can define a quasi-local algebra $\Bl_\s$ in the following way:
\begin{eqnarray}\label{Bsigma}
\Bl_\s&:=&\Ot_{j\geq 0}\left(\Ot_{x\in \overrightarrow{W_j}}\Bl_{P_{\s(x)}}\right)\no
&\equiv& B(H_{\s(x_0),0})\ot\left(\Ot_{j\geq 0}\left(B(H_{\s(x),1})\ot
B(H_{\s(\overrightarrow{S(x)}),0})\right)\right)
\end{eqnarray}
A completely positive identity preserving map $E_\s:\Bl_L\mapsto
\Bl_\s$ is defined as the (infinite) tensor product of the mappings
\begin{equation}\label{Esigma}
a\in \Bl_x\mapsto P_{\s(x)}a P_{\s(x)}
\end{equation}
The map $E_\s$ satisfies $E_\s=E_\s\circ E$, where $E$ is defined by
\eqref{E}.

Define a state $\psi_\s$ on $\Bl_\s$ as follows:
\begin{equation}\label{psi}
\psi_{\s_{\lceil \Lb_n}
}:=\eta^{(0)}_{\s(x_0)}\ot\left(\bigotimes_{j\geq
0}\left(\bigotimes_{x\in W_j}
\eta^{(j)}_{\s(x),\s(\overrightarrow{S(x)})}\right)\right)\ot
\bigg(\bigotimes_{y\in W_n}\widehat{\eta}^{(n)}_{\s(y)}\bigg)
\end{equation}
where the state $\eta^{(0)}_{\s(x_0)}$ is defined on
$B(H_{\s(x_0),0})$ by
\begin{equation}
\eta^{(0)}_{\s(x_0)}(a):=\frac{\ffi\left(i_{\Lb_0}\left(P_{\s(x_0)}(a\ot
\id_{H_{\s(x),1}})P_{\s(x_0)}\right)\right)}{\ffi(P_{\s(x_0)})}
\end{equation}
For $x\in W_j$,  the state $\eta^{(j)}_{\s(x)}$ is defined on
$B(H_{\s(x),1})\ot B(H_{\s(\overrightarrow{S(x)}),0})$ as follows:
\begin{equation}
\eta^{(j)}_{\s(x),\s(\overrightarrow{S(x)})}(a\ot b):=\frac{\phi_{\s(x)}\left(a\ot
i_{\Lb_{\overrightarrow{S(x)}}}\left(P_{\s(\overrightarrow{S(x)})}(b\ot
\id_{H_{\s(\overrightarrow{S(x)}),1}})P_{\s(\overrightarrow{S(x)})}\right)\right)}{\pi_{\s(x),\s(\overrightarrow{S(x)})}}
\end{equation}
And for $y\in W_n$, the state $\widehat{\eta}^{(n)}_{\s(y)}$ is
defined on $B(H_{\s(y),1})$ as follows:
\begin{equation}
\widehat{\eta}^{(n)}_{\s(y)}(a):=\sum_{\s'\in
\Og_{\overrightarrow{S(x)}}}\eta^{(n)}_{\s(x),\s'(\overrightarrow{S(x)})}
\end{equation}

Let $\left\{\El^{j}_{\s(x)}\right\}_{j\geq 0,  x\in W_j}$ be a
family of transition expectations:
$$\El^{j}_{\s(x)}: \Bl_{P_{\s(x)}}\ot \Bl_{P_{\s(\overrightarrow{S(x)})}}\mapsto \Bl_{P_{\s(x)}}$$
defined by
\begin{eqnarray}
&&\El^{j}_{\s(x)}\left((a_0\ot a_1)\ot(b_0\ot(\ot_{y\in
\overrightarrow{S(x)}} b_{y,1}))\right)\no
&&=\left(\eta^{(j)}_{\s(x),\s(\overrightarrow{S(x)})}(a_1\ot
b_0)(\prod_{y\in
\overrightarrow{S(x)}}\eta^{(j+1)}_{\s(y),\s(\overrightarrow{S(y)})}
(b_{y,1}\ot\id_{H_{\overrightarrow{S(y)},0}}))\right)a_0\ot\id
\end{eqnarray}
And let $\El^j_{\s_{\lceil W_j}}:\Bl^{j}_{\s(W_j)}\ot
\Bl^{j+1}_{\s(W_{n+1}}\to \Bl^{j}_{\s(W_j)}$ be a transition
expectation defined by:
\begin{equation}\label{Elsigma}
\El^j_{\s_{\lceil{W_j}}}:=\Ot_{x\in W_j}\El^j_{\s(x)}.
\end{equation}

One can prove the following fact.

\begin{proposition}
The state $\psi_\s$ satisfies the Markov property w.r.t the family
of transition expectations $\left\{\El^{j}_{\s}\right\}_{j\geq 0}$
given by $\eqref{Elsigma}$.
\end{proposition}

\begin{theorem}\label{disintegration}
Let $\ffi$ be a Markov state w.r.t the family of localized
transitions expectations $\{\El^{j} \}_{j\geq 0 }$. Assume that
$\Og$ and $\mu$ are as in Section 5 and the quasi-local algebra
$\Bl_\s$ is given by $\eqref{Bsigma}$, the map $E_\s$ by tensor of
the maps $\eqref{Elsigma}$; the state $\psi_\s$ on $\Bl_\s$ is
defined by $\eqref{psi}$. Then the state $\ffi$ admits a
disintegration $$\ffi=\int_{\Og}\ffi_{\s}\mu(d\s),$$ where $\s\in\Og
\mapsto \ffi_\s\in \mathcal{S}(\Bl_L)$ is a
$\s(\Bl_L^*,\Bl_L)$-measurable map satisfying, for $\mu$-almost all
$\s\in \Og$,
$$\ffi_\s=\psi_\s\circ E_\s$$
\end{theorem}

\begin{proof}
Let $\ffi$ be a Markov state w.r.t the sequence $\{\El_j\}_{j\geq
0}$. We define a commutative $C^{*}-$subalgebra $\mathcal{C}$
defined by:
$$\mathcal{C}:=\Ot_{j \geq 0}\left(\Ot_{x\in W_j}Z^{j}_x \right)$$
Following \cite{AF03}, let $(H,\pi)$ be the GNS representation of
$\Rl$ associate with $\ffi_{\lceil \Rl}$. Then
$\pi(\mathcal{C})''\subset \pi(\Rl)'\cap \pi(\Rl)''$ (see \cite[
section III.2, Theorem 7.2]{T}) hence, the representation $\pi$ can
be disintegrated as follows:
$$\pi= \int_{\Og}^{\oplus}\pi_{\s}\mu(d\s)$$
where $\s\mapsto \pi_{\s}$ is a weakly measurable field of
representation of $\Rl$ (see \cite[Theorem IV.8.25]{T},\cite{AF03}).
Furthermore one can find a measurable field $\s\mapsto \zeta_\s$ of
unit vectors such that, for each $a\in \Rl$, we get
$$\ffi(a)=\int_{\Og}< \pi_{\s}(a)\zeta_\s, \zeta_{\s}>\mu(d\s)$$
As $\ffi$ is a Markov state, it's invariant under $E$ and the previous expression can be extended on
$\Bl_L$ by
$$\ffi=\int_{\Og}\ffi_{\s}\mu(d\s)$$
for the $\s(\Bl_L^*,\Bl_L)$-measurable field $\ffi_\s$ defined as
$$\ffi_\s:= < \pi_\s(E(\cdot))\zeta_\s,\zeta_\s>.$$
Let us now prove the second part of the theorem about the expression
of $\ffi_\s$. For $\bar{\s}\in \Og$ take an element $A\in
\Bl_{\Lb_n}$ given by
$$A=\Ot_{j\geq 0}\left(\Ot_{x\in W_j}P_{\bar{\s}(x)}(a_{\bar{\s}(x),0}\ot
a_{\bar{\s}(x),1})P_{\bar{\s}(x)}\right)$$
and consider $$Z=\Ot_{j\geq 0}\left(\Ot_{x\in W_j}P_{\bar{\s}(x)}\right)$$
Then we obtain $$\ffi(i_{\Lb_n}(AZ))=\int_{\Og}z(\s)\ffi_\s(i_{\Lb_n}(AZ))\mu(d\s)$$
where $$z(\s)=\delta_{\s_{\lceil W_n},\bar{\s}_{\lceil W_n}}$$

Now for $j\in\{0,\cdots,n \}$, define
$$A_{j}=\Ot_{x\in
W_j}P_{\s(x)}(a_{\s(x),1}\ot a_{\s(x),1})P_{\s(x)}.$$ Using the
Markov property, $\ffi(i_{\Lb_n}(AZ))$ can be computed as follows:
\begin{eqnarray*}
\ffi(i_{\Lb_n}(AZ))=z(\s)\ffi(i_{\Lb_n}(A))&=&z(\s)\sum_{\s_{n+1}\in
\Og_{W_{n+1}}}\ffi(i_{\Lb_{n+1}}(A\ot P_{\s_{n+1}}))\no
&=&z(\s)\sum_{\s_{n+1}\in\Og_{n+1}}\ffi(\El^{(0)}(A_0\ot
\El^{(1)}(A_1\ot\cdots \\[2mm]
&&\ot \El^{(n-1)}(A_{n-1}\ot \El^{(n)}(A_{n}\ot P_{\s_{n+1}})))))
\end{eqnarray*}
According to the localized property of $\El^{(j)}$ one has
\begin{eqnarray*}
\El^{(n)}(A_n\ot P_{\s_{n+1}})&=&\Ot_{x\in
W_n}\El_{x}^{(n)}(P_{\s(x)}(a_{\s(x),0}\ot a_{\s(x),1})P_{\s(x)}\ot
P_{\s_{n+1}(\overrightarrow{S(x)})})\no &=&\Ot_{x\in
W_n}\phi_{\s(x)}\left(P_{\s(x)}a_{\s(x)}P_{\s(x)}\ot
P_{\s_{n+1}(\overrightarrow{S(x)})}\right)P_{\s(x)}(a_{\s(x),0}\ot
\id_{\s(x),1})\no &=&\prod_{x\in
W_n}\phi_{\s(x)}\left(P_{\s(x)}a_{\s(x)}P_{\s(x)}\ot
P_{\s_{n+1}(\overrightarrow{S(x)})}\right) \Ot_{x\in W_n}
P_{\s(x)}(a_{\s(x),0}\ot \id_{\s(x),1})\no &=&\prod_{x\in
W_n}\eta^{(n)}_{\s(x),\s_{n+1}(\overrightarrow{S(x)})}(a_{\s(x),1}\ot
\id_{\s_{n+1}(\overrightarrow{S(x)}),0})\pi_{\s(x),\s_{n+1}(\overrightarrow{S(x)})}\no
\end{eqnarray*}
Let $j\in\{0,\cdots, n-1\}$ and denote
$$A_{j+1,0}=\Ot_{x\in
W_{j+1}}P_{\s(x)}(a_{\s(x),0}\ot \id_{s(x),1})P_{\s(x)}$$ and
$a_{\overrightarrow{S(x)},0}=\Ot_{y\in
\overrightarrow{S(x)}}a_{\s(y),0}$. Then we get
\begin{eqnarray*}
\El^{(j)}\left(A_j\ot A_{j+1,0}\right)&=&\Ot_{x\in W_n}\El_{x}^{(n)}\left(P_{\s(x)}(a_{\s(x),0}\ot
a_{\s(x),1})P_{\s(x)}\ot P_{\s(\overrightarrow{S(x)})}(a_{\s(\overrightarrow{S(x)},0)}\ot
\id_{\overrightarrow{S(x)},1} )P_{\s(\overrightarrow{S(x)})}\right)\\
&=&\Ot_{x\in W_j}\phi_{\s(x)}\left(P_{\s(x)}a_{\s(x),1}P_{\s(x)}\ot
P_{\s(\overrightarrow{S(x)})}\right) P_{\s(x)}(a_{\s(x),0}\ot
\id_{\s(x),1})\no &=&\prod_{x\in
W_j}\phi_{\s(x)}\left(P_{\s(x)}a_{\s(x),1}P_{\s(x)}\ot
P_{\s(\overrightarrow{S(x)})}(a_{\s(\overrightarrow{S(x)}),0}\ot
a_{\s(\overrightarrow{S(x)}),1})P_{\s(\overrightarrow{S(x)})}\right)\no
&&\times \Ot_{x\in W_n} P_{\s(x)}(a_{\s(x),0}\ot \id_{\s(x),1})\no
&=&\prod_{x\in
W_j}\eta^{(j)}_{\s(x),\s(\overrightarrow{S(x)})}(a_{\s(x),1}\ot
\id_{\s(\overrightarrow{S(x)}),0})\pi_{\s(x),\s(\overrightarrow{S(x)})}
\Ot_{x\in W_n} P_{\s(x)}(a_{\s(x),0}\ot \id_{\s(x),1})\no
\end{eqnarray*}

Hence,
\begin{eqnarray*}
\ffi(i_{\Lb_n}(AZ))&=z(\s)&\ffi\left(\sum_{\s_{n+1}\in \Og_{n+1}}\prod_{j=1}^{n-1}\prod_{x\in
W_j}\eta^{(j)}_x(a_{\s(x),0}\ot a_{\s(\overrightarrow{S(x)}),0})\cdot
\pi_{\s(x),\s(\overrightarrow{S(x)})}\times a_{\s(x_0),0}\ot \id_{\s(x_0),1}\right)\no
&=&z(\s) \sum_{\s_{n+1}\in \Og_{W_{n+1}}}\psi_{\s\vee \s_{n+1}}(A\ot \id_{W_{n+1}})\no
&=&\int_{\Og}z(\s) \psi_{\s}\circ E_\s(A) \mu(d\s)
\end{eqnarray*}
Consequently, one has
$$\int_{\Og}z(\s)\ffi_{\s}\mu(d\s)\, =\, \int_{\Og}z(\s)\psi_{\s}(E_{\s}(a))\mu(d\s)$$
for each fixed localized operator $A\in \Bl_L$ and each function
$z\in C(\Og)$ depending only on finitely many variables. As such
functions are dense on $C(\Og)$ and due to the uniqueness of the
Radan-Nikodym derivative, for each localized $A\in\Bl_L$, there
exists a measurable set $\Og_A\subset \Og$ of full $\mu-$measure
such that for $\s\in \Og_0$ one has:
\begin{equation}\label{fffi}
\ffi_\s(A)=\psi_{\s}(E_\s(A))
\end{equation}
By considering linear combination with rational coefficients,  one
can find  dense subset $\Bl_{L,0}$  of localized operators such that
\eqref{fffi} still satisfied for each element of $\Bl_{L,0}$.

 Let
$a\in \Bl_L$ and $(a_n)_n$ be a sequence of $\Bl_{L,0}$ that
converges to $a$. Then for $\s\in \Og$ one has:
$$\ffi_{\s}(a)\;=\; \lim_{n}\; \ffi_{\s}(a_n)\; = \; \lim_{n}\;\psi_{\s}(E_{\s}(a_n))\; =
\;\psi_{\s}(E_{\s}(a))$$
 Thus the equation $\eqref{fffi}$ holds for all $\a\in \Bl_L$.
\end{proof}
 
\begin{corollary}
Let $$\ffi(a)=\int_{\Og}\ffi_{\s}\mu(d\s)$$ be the disintegration of
the Markov state $\ffi$ as in theorem $\eqref{disintegration}$. Then
$\phi_\s$ is a factor state for $\mu-$almost all $\s\in \Og$.
\end{corollary}

The proof is similar to \cite[Corollary 3.3]{AF03}.

\section{A reconstruction theorem}

In this section we study the converse direction of the disintegration result in the previous
section.

Let us consider for $j\geq 0$ and $x\in W_j$,  a commutative
subalgebra $\mathcal{Z}^j_x$ of $\Bl_x$ with spectrum its $\Og_{x}^j$ together
with its family of projections $\{P_{\s(x)}^j\}_{\s(x)\in
\Og^{j}_{x}}$. For $\s\in \Og$ we assume that the following
distributions are given:
\begin{itemize}
  \item $\pi_{\s(x_0)}^{0}>0$;  \quad  ({\it initial distribution})
\item For $j \geq 0$ and $ x\in W_j \; \quad \pi^{j}_{\s(x),\s(\overrightarrow{S(x)})}>0$; \quad
({\it transition probabilities})
\end{itemize}
such that for $n\geq 0$:  $$\sum_{\s'\in \Og_{W_{n+1}}}\prod_{x\in
W_j}\pi^{j}_{\s(x),\s'(\overrightarrow{S(x)})}=1$$ Then a Markov
measure $\mu$  on $\Og:=\prod_{j\geq 0}\prod_{x\in W_j}\Og_x^{j}$ is
defined as follows: For $n\geq 1$ and $\s\in \Og_{\Lb_n}$:
$$\mu(\s)=\pi_{\s(x_0)}^{0}\prod_{j=0}^{n-1}\prod_{x\in
W_j}\pi^{j}_{\s(x),\s(\overrightarrow{S(x)})}$$
Now for $j\geq 0$ and $x\in \Bl^{j}_x$ and $P_{\s(x)}$ a given central projection we set
$$\Bl^{j}_{P_{\s(x)}}:=P_{\s(x)}\Bl_{x}^{j}P_{\s(x)}=N^{j}_{\s(x)}\ot \overline{N}^{j}_{\s(x)},$$
where $N^{j}_{\s(x)}$ and
$\overline{N}^{j}_{\s(x)}$ are finite dimensional factors.

For $\s\in \Og$, let $\eta_{\s(x_0)}^{0}$ be a state on
$N_{\s(x_0)}$ and $\eta_{\s(x),\s(\overrightarrow{S(x)})}^{j}$ be a
state on $\overline{N}^{j}_{\s(x)}\ot N^{j+1}_{\s(S(x))}$ with
$N^{j+1}_{\s(S(x))}=\Ot_{y\in W_{n+1}}N_{\s(y)}$. For each $\s\in
\Og$ we define the state $\psi_{\s}$ by $\eqref{psi}$ on the
quasi-local algebra $\Bl_{\s}$ defined by $\eqref{Bsigma}$. Let
$E_{\s}:\Bl_L \to \Bl_\s$ be given in $\eqref{Esigma}$, together
with the $\s(\Bl_L^*, \Bl_L)-$measurable map
$$\s\in \Og \mapsto \ffi_{\s}:=\psi_{\s}\circ E_{\s}$$
\begin{theorem}
For the same notations as above, the state $\ffi$ on $\Bl_L$ given by
$$\ffi:=\int_{\Og}\ffi_{\s}\mu(d\s)$$
is a Markov state w.r.t the sequence of transition expectations
$\{\El^{j}\}_{j\geq 0}$, determined by the states $\phi_{\s(x)}^{j}$
satisfying, for each $j\geq0, x\in W_{j}$ and $\s\in \Og$ the
following equality
\begin{eqnarray}
&&\phi^{j}_{\s(x)}(\overline{a}\ot P^{j+1}_{\s(\overrightarrow{S(x)})}(b\ot
\overline{b})P^{j+1}_{\s(\overrightarrow{S(x)})})\no
&&=\sum_{\s_{j+2}\in
\Og_{W_{j+2}}}\pi^{j}_{\s(x),\s(\overrightarrow{S(x)})}\eta^{j}_{\s(x),\s(\overrightarrow{S(x)})}
(\overline{a}\ot b)\prod_{y\in \overrightarrow{S(x)}}\pi^{j+1}_{\s(y),\s(\overrightarrow{S(y)})}
\eta^{j+1}_{\s(y),\s(\overrightarrow{S(y)})}(\overline{b_{y}}\ot \id)
\end{eqnarray}
with $a\in N^{j}_{\s(x)},\; \overline{a}\in
\overline{N}^{j}_{\s(x)}, \; b=\ot_{y\in S(x)}b_{\s(y)}\in
N^{j+1}_{\overrightarrow{S(x)}}$ and $\overline{b}=\ot_{y\in
\overrightarrow{S(x)}}\overline{b}_{\s(y)}\in
\overline{N}_{\overrightarrow{S(x)}}$.
\end{theorem}
\begin{proof}
For $\s\in\Og$ the state $\psi_\s$ is well-defined, in addition it's
a quantum Markov state w.r.t the sequence
$\{\El^j_{\s_{\lceil{W_j}}}\}_{(j\geq 0)}$ of the transition
expectations given by $\eqref{Elsigma}$.

Let $n\geq 0$ and take an element
$$A=\Ot_{j=0}^{n}\Ot_{x\in
W_j}P^{j}_{\overline{\s}(x)}(a_{\overline{\s}(x)}\ot
\overline{a}_{\overline{\s}})P^{j}_{\overline{\s}(x)}$$ in
$\Bl_{\Lb_n}$. Using $\eqref{caracofmarkovstate}$ it is enough to
proof that for such an element $A$ the equality
 $$\ffi(A)=\ffi\left(B\right)
$$
with
\begin{eqnarray*}
&&B=\El^{(0)}\left(A_{W_0}\ot\cdots \ot
A_{W_{n-1}}\ot\El^{(n)}\left(A_{W_{n-1}}\ot
A_{W_{n}}\right)\cdots\right)\\[2mm]
&&A_{W_j}=\Ot_{x\in
W_j}P^{j}_{\overline{\s}(x)}(a_{\overline{\s}(x)}\ot
\overline{a}_{\overline{\s}})P^{j}_{\overline{\s}(x)}.
\end{eqnarray*}
For $\s\in \Og_{\Lb_n}$ and $E_\s(A)=\delta_{\s_{\lceil \Lb_{n}},
\overline{\s}_{\lceil \Lb_{n}}}A$

$$\ffi_\s(A)=\ffi_\s\left(\El_\s^{(0)}\left(A_{\s(W_0)}\ot\cdots \ot
A_{\s(W_{n-1})}\ot\El_{\s}^{(n-1)}\left(A_{\s(W_{n-1})}\ot
A_{\s(W_{n})}\right)\cdots\right)\right),$$
where $A_{\s(x)}=a_{\s(x)}\ot \overline{a}_{\s(x)}.$

One has
\begin{eqnarray*}
\ffi_\s(i_{\Lb_n}(A))&=&\delta_{\s_{\lceil \Lb_n},
\overline{\s}_{\lceil \Lb_n}} \sum_{\s'\in
\Og_{W_{n+1}}}\psi_\s\left(\Ot_{j=1}^{n}\Ot_{x\in
W_j}(P^j_{\s(x)}(a_{\s(x)}\ot
\overline{a}_{\s(x)})P^j_{\s(x)})\right) \\[2mm]
&&\times\prod_{y\in
W_n}\eta_{\s(y),\s'(\overrightarrow{S(y)})}(\overline{a}_{\s(y)}\ot
\id_{\s'(\overrightarrow{S(y)}})\no\no &=&\delta_{\s_{\lceil \Lb_n},
\overline{\s}_{\lceil \Lb_n}}\sum_{\s'\in
\Og_{W_{n+1}}}\eta^{(0)}_{\s(x_0)}(a_{\s(x_0)})\left(\prod_{j=0}^{n-1}\prod_{x\in
W_j}\eta^{(j)}_{\s(x),\s(\overrightarrow{S(x)})}(\overline{a}_{\s(x)}\ot
a_{\s(\overrightarrow{S(x)})})\right)\no &&\times\prod_{y\in
W_n}\eta_{\s(y),\s'(\overrightarrow{S(y)})}(\overline{a}_{\s(y)}\ot
\id_{\s'(\overrightarrow{S(y)}})\no
\end{eqnarray*}
Then 
\begin{eqnarray*}
\ffi(i_{\Lb_n}(A))&=&\int_{\Og}\ffi_{\s}(i_{\Lb_{n}}(A))\mu(d\s)\no
&=& \sum_{\s'\in
\Og_{W_{n+1}}}\pi_{(x_0)}^{0}\eta^{(0)}_{(x_0)}(a_{\sb(x_0)})\prod_{j=0}^{n-1}\prod_{x\in
W_j}\pi^{j}_{\sb(x),\sb(\overrightarrow{S(x)})}
\eta^{(j)}_{(x),\sb(\overrightarrow{S(x)})}(\overline{a}_{ (x)}\ot a_{ (\overrightarrow{S(x)})})\no
&&\times\prod_{y\in W_n}\pi^{n}_{ (x), (\overrightarrow{S(x)})}\eta_{
(y),\s'(\overrightarrow{S(y)})}(\overline{a}_{\sb(y)}\ot \id_{\s'(\overrightarrow{S(y)}})\no
\end{eqnarray*}
Conversely we have
\begin{eqnarray*}
\El_{\s}^{(n-1)}\left(A_{\sb(W_{n-1})}\ot
A_{\sb(W_{n})}\right)&=&\delta_{\s_{\lceil \Lb_n},
\overline{\s}_{\lceil \Lb_n}}\Ot_{x\in
W_{n-1}}\El_{\s(x)}\left(a_{\s(x)}\ot \overline{a}_{\s(x)}\ot(
a_{\s(\overrightarrow{S(x)})}\ot
\overline{a}_{\s(\overrightarrow{S(x)})})\right)\no
&=&\delta_{\s_{\lceil \Lb_n}, \overline{\s}_{\lceil \Lb_n}}
\Ot_{x\in W_{n-1}}\sum_{\s'\in
\Og_{\overrightarrow{S(x)}}}\phi^{(n-1)}_{\s(x)}(\overline{a}_{\s(x)}\ot
(a_{\s(\overrightarrow{S(x)})}\ot \id))\no &&\times\prod_{y\in
\Og_{W_n}}\phi^{(n)}_{\s(y)}(\overline{a}_{\s(y)}\ot
\id_{\s'(\overrightarrow{S(y)})})(P^{n}_{\s(x)}a_{\s(x)}\ot \id
P^{n}_{\s(x)}\no &=&\delta_{\s_{\lceil \Lb_n}, \overline{\s}_{\lceil
\Lb_n}} \sum_{\s'\in \Og_{W_{n+1}}}\prod_{x\in
W_{n-1}}\phi^{(n-1)}_{\s(x)}(\overline{a}_{\s(x)}\ot
(a_{\s(\overrightarrow{S(x)})}\ot \id))\\[2mm]
&&\times\prod_{y\in
\overrightarrow{S(x)}}\phi^{(n)}_{\s(y)}(\overline{a}_{\s(y)}\ot
\id_{\s'(\overrightarrow{S(y)})})\Ot_{x\in
W_n}P^{n}_{\s(x)}(a_{\s(x)}\ot \id)P^{n}_{\s(x)}\no
\end{eqnarray*}

After small computation, we get
\begin{eqnarray*}
\ffi_\s(B)&=&\delta_{\s_{\lceil \Lb_n}, \overline{\s}_{\lceil
\Lb_n}}\sum_{\s'\in
\Og_{W_{n+1}}}\phi^{(0)}_{\s(x_0)}(a_{\s(x_0)})\prod_{j=0}^{n-1}\prod_{x\in
W_j}\phi^{(j)}_{\s(x)}(\overline{a}_{\s(x)}\ot
(a_{\s(\overrightarrow{S(x)})}\ot \id))\\[2mm]
&&\times\prod_{y\in W_{n}}\phi^{(n)}_{\s(y)}(\overline{a}_{\s(y)}\ot
\id_{\s'(\overrightarrow{S(y)})})\no
\end{eqnarray*}
Then by taking into account that
$\phi^{(0)}_{\s(x_0)}(\cdots)=\pi^0_{\s(x_0)}\eta^{(0)}_{\s(x_0)}(\cdot)$
in addition to the expression of
$\phi^{j}_{\s(x),\s(\overrightarrow{S(x)})}$ given in the theorem,
we obtain
\begin{eqnarray*}
\ffi(i_{\Lb_{n}}(B))&=&\int_{\Og}\ffi_\s(B)\mu(d\s)\no
&=&\sum_{\s'\in \Og_{W_{n+1}}}\phi^{(0)}_{(x_0)}(a_{\s(x_0)})\prod_{j=0}^{n-1}\prod_{x\in
W_j}\left(\phi^{(j)}_{(x)}(\overline{a}_{(x)}\ot (a_{(\overrightarrow{S(x)})}\ot
\id))\right)\prod_{y\in W_{n}}\phi^{(n)}_{(y)}(\overline{a}_{(y)}\ot
\id_{\s'(\overrightarrow{S(y)})})\no
&=&\sum_{\s'\in
\Og_{W_{n+1}}}\pi^0_{(x_0)}\eta^{(0)}_{(x_0)}(a_{\s(x_0)})\prod_{j=0}^{n-1}\prod_{x\in
W_j}\left(\pi^{(j)}_{(x), (\overrightarrow{S(x)})}\eta^{(j)}_{(x),
(\overrightarrow{S(x)})}(\overline{a}_{(x)}\ot a_{ (\overrightarrow{S(x)})})\right)\no
&&\times\prod_{y\in W_n}\pi^{(n)}_{
(y),\s'(\overrightarrow{S(y)})}\eta^{(n)}_{\sb(y),\s'(\overrightarrow{S(y)})}(\overline{a}_{ (y)}\ot
\id_{\s'(\overrightarrow{S}(y))})\no
&=&\ffi(i_{\Lb_{n}}(A)).
\end{eqnarray*}
Hence, the proof is complete.

\end{proof}
\section{Connection with statistical mechanics}

In this section we study the link between Markov states on the
Cayley tree and the Ising  potentials through the Markov property.

Let us assume that we have a locally faithful Markov state $\ffi$ on
the quasi-local algebra $\Bl_{L}$, then a potential $h_\Lb$  is
canonically defined for each finite subset $\Lb\subset L$ as
follows:
\begin{equation}\label{potential}
\ffi_{\lceil \Bl_\Lb}\; = \;Tr_{\Bl_\Lb}(e^{-h_\Lb} \cdot)
\end{equation}
The set of potentials $\{h_\Lb\}_{\Lb\subset_{fin} L}$ satisfy
normalization conditions $$\Tr_{\Bl_\Lb}(e^{-h_\Lb})\; =\; 1 $$
together with compatibility conditions $$(\Tr_{\Bl_{\Lb'\setminus
\Lb}}\ot \id_{\Bl_\Lb})(e^{-h_{\Lb'}})\; =\; e^{-h_\Lb} $$
 for finite subsets $\Lb\subset\Lb'\subset_{fin}L$.
 In particular for each $n\geq 0$, one has
  $$(\Tr_{\Bl_{W_{n+1}}}\ot \id_{\Bl_{\Lb_n}})(e^{-h_{\Lb_{n+1}}})\; =\; e^{-h_{\Lb_n}}. $$

 \begin{theorem}
 Let $\ffi$ be a locally faithful state on $\Bl_L$. Then the following assertions are
 equivalent:
 \begin{enumerate}
 \item[(i)] $\ffi$ is a Markov state  w.r.t. the localized sequence
 $\{\El^{(j)}\}_{j\geq0}$ of transition expectations;
\item[(ii)] The sequence of potentials $\{h_{\Lb_n}\}$ associated to $\ffi$ by $\eqref{potential}$, can be recovered by
\begin{equation}\label{potentialequation}
 h_{\Lb_n}= H_{W_0} + \sum_{j=0}^{n-1}H_{W_j,W_{j+1}} + \Ht_{W_n}
 \end{equation}
where the sequences $\{H_{W_{j}}\}_{j\geq
0},\;\{\Ht_{W_{j}}\}_{j\geq 0}$
and $\{H_{W_{j}, W_{j+1}}\}_{j\geq 0}$ of self-adjoint operators localized in $\Bl_{W_j}$ and
$\Bl_{\Lb_{j,j+1}}$,
 respectively, and satisfying commutation relations
\begin{eqnarray}\label{commutation}
&&[H_{W_n}, H_{W_n,W_{n+1}}]=0,\quad  [ H_{W_n,W_{n+1}},
\Ht_{W_{n+1}}]=0,\nonumber\\ &&[H_{W_n}, \widehat{H}_{W_n}]=0,\quad
[H_{W_n,W_{n+1}}, H_{W_{n+1},W_{n+2}}]=0.
 \end{eqnarray}
 \end{enumerate}
 \end{theorem}

\begin{proof} $(i)\Rightarrow (ii)$. Let $\varphi$ be a locally faithful Markov state w.r.t. the
sequence
 $\{\El^{j}\}_{j\geq0}$ of transition expectations. For every $j\geq 0$, $x\in W_j$ and $\sigma\in
\Omega$ we define the following set of potentials $\{h^{0}_{\sigma(x_0)}\}, \{h^{j}_{\sigma(x),
\sigma(\overrightarrow{S(x)})}\}$ and
$\{\hat{h}^{(j)}_{\sigma(x)}\}$ related to the following positive
functionals. Namely, the potential $h^{0}_{\sigma(x_0)}$ is related
to $\pi^0_{\sigma(x_0)}\eta^{0}_{\sigma(x_0)}$ on $N^{0}$, the
potential   $h^{j}_{\sigma(x), \sigma(\overrightarrow{S(x)})}$ is
related to
$\pi^{j}_{\s(x),\s(\overrightarrow{S(x)}}\eta^{j}_{\s(x),\s(\overrightarrow{S(x)}}$
on $\overline{N}^{j}_{\s(x)}\ot N^{j+1}_{\s(\overrightarrow{S(x)})}$
and $\{\hat{h}_{\sigma(x)}\}$ is related to $$\sum_{\s'\in
\Og_{\overrightarrow{S(x)}}
}\pi^{j}_{\s(x),\s'(\overrightarrow{S(x)})}\eta^{j}_{\s(x),
\s'(\overrightarrow{S(x)})}(\cdot\ot \id)$$ on
$\overline{N}^{j}_{\s(x)}$.
 The potential $\{\hat{h}^{(j)}_{\sigma(x)}\}$ is related to
$$\sum_{\s'\in
\Og_{\overrightarrow{S(x)}}}\pi^{j}_{\s(x),\s'(\overrightarrow{S(x)})}\eta^{j}_{\s(x),\s'(\overrightarrow{S(x)})}$$
Take any localized element
$$a=\Ot_{j=0}^{n}\ot_{x\in W_j}a_{\s(x)}\ot \overline{a}_{\s(x)}$$
from $\Bl_{\s(\Lb_n)}$. Then one has
\begin{eqnarray*}
\varphi_{\s}(a)&=&\psi_{\s}(a)\no &=&
\eta^{0}_{x_0}(a_{\s(x_0)})\prod_{j=0}^{n-1}\prod_{x\in
W_j}\eta^{j}_{\s(x),
\s(\overrightarrow{S(x)})}(\overline{a}_{\s(x)}\ot
a_{\overrightarrow{S(x)}}) \prod_{y\in W_n}\eta^{j}_{\s(y),
\s(\overrightarrow{S(y)})}(\overline{a}_{\s(x)}\ot
\id_{\overrightarrow{S(x)}})\no &=&\Tr
(\textrm{e}^{-h^{0}_{\s(x_0)}}a_{\s(x_0)})\prod_{j=0}^{n-1}\prod_{x\in
W_j}\Tr(\textrm{e}^{-h^{j}_{\s(x),
\s(\overrightarrow{S(x)})}}\overline{a}_{\s(x)}\ot
a_{\overrightarrow{S(x)}})\prod_{y\in
W_n}\Tr(\textrm{e}^{-\widehat{h}^{n}_{\s(y)}}\overline{a}_{\s(x)})
\end{eqnarray*}

While in the last expression the traces are taken on disjoint
tensors,  we then get
$$\varphi_{\s}(a)= \Tr(\textrm{e}^{-h_{\s_{\lceil\Lb_n}}}a)$$
with
$$h_{\s_{\lceil\Lb_n}}=h^{0}_{\s(x_0)}+\sum_{j=0}^{n-1}\sum_{x\in
W_j}h^{j}_{\s(x),\s(\overrightarrow{S(x)})}+ \sum_{y\in
W_n}\widehat{h}^{n}_{\s(y)}$$ Then we define
 \begin{eqnarray*}
 &&H_0:=\sum_{\s\in \Og_{W_j}}P^{j}_{\s(x)}(h^{j}_{\s(x)}\ot \id)P^{j}_{\s(x)}\no
&&H_{x,\overrightarrow{S(x)}}:=\sum_{\s\in \Og_{\{x,\overrightarrow{S(x)}\}}}(P^{j}_{\s(x)}\ot
P^{j+1}_{\s(\overrightarrow{S(x)})})(\id\ot h^{j}_{\s(x),\s(\overrightarrow{S(x)})}
\ot \id)(P^{j}_{\s(x)}\ot P^{j+1}_{\s(\overrightarrow{S(x)})})\no
   &&\widehat{H}_x:=\sum_{\s\in \Og_{x}}P^{j}_{\s(x)}(\id\ot \widehat{h}^{j}_{\s(x)})P^{j}_{\s(x)}
 \end{eqnarray*}

 and
  \begin{eqnarray*}
  &&H_{j,j+1}=\sum_{x\in W_j}H_{x,\overrightarrow{S(x)}}\no
   &&\widehat{H}_j:=\sum_{x\in W_{j}}\widehat{H}_x
 \end{eqnarray*}
 Finally we set $$h_{\L_n}=H_0 + \sum_{j=0}^{n-1}H_{j,j+1}+ \widehat{H}_{n}.$$
 Then $h_{\L_n}$ is the  potential related to the state $\varphi$ on $\Bl_{\L_n}$.
$(ii)\Rightarrow (i)$. For $ n\in \mathbf N, x\in W_n$ we consider the map
$\El^{(n)}_{x}:\Bl_{x}\ot \Bl_{\overrightarrow{S(x)}}\to \Bl_{x}$
defined by
$$\El^{(n)}_x(a)=(\Tr_{\overrightarrow{S(x)}]}\ot \id_{x})(A^{*}_{x, \overrightarrow{S(x)}}a
A_{x,\overrightarrow{S(x)}})$$
with $$A_{x,\overrightarrow{S(x)}}=\textrm{e}^{-\frac{1}{2}H_{x,
\overrightarrow{S(x)}}}\textrm{e}^{-\frac{1}{2}\widehat{H}_{
\overrightarrow{S(x)}}}\textrm{e}^{\frac{1}{2}\widehat{H}_{x}}$$
considering $K_{n,n+1}=\Ot_{x\in W_n}A_{x,\overrightarrow{S(x)}}$
(or also by taking $K_{n,n+1}=\textrm{e}^{-\frac{1}{2}H_{n,
n+1}}\textrm{e}^{-\frac{1}{2}\widehat{H}_{
n+1}}\textrm{e}^{\frac{1}{2}\widehat{H}_{n}}$). We get a transition
expectation
$$\El^{(n)}(a)=(\Tr_{W_{n+1}]}\ot \id_{W_n})(K^{*}_{n, n+1}a K_{n,n+1})$$
and its corresponding quasi-conditional expectation is defined
by
$$E^{n}=\id_{\Bl_{\Lb_{n-1}}}\ot \El^{(n)}.$$
One can easily check that $\varphi$ is a Markov state w.r.t the
sequence $\{E^{n}\}$ of conditional expectations.
\end{proof}

\section*{Acknowledgments}
The authors are grateful to professors L. Accardi and F. Fidaleo for
their fruitful discussions and useful suggestions.

\end{document}